\documentclass[11pt]{article}
\usepackage{fullpage}
\usepackage[numbers]{natbib}
\usepackage[utf8]{inputenc}
\usepackage{amsmath,amsthm,amssymb}
\usepackage{thmtools} 
\usepackage{thm-restate}
\usepackage{amsfonts,bm,bbold,dsfont}
\usepackage{enumitem}
\usepackage[normalem]{ulem}
\usepackage{xspace}
\usepackage[usenames,dvipsnames]{xcolor}
\usepackage[colorlinks=true,citecolor=ForestGreen,linkcolor=ForestGreen,urlcolor=NavyBlue]{hyperref}

\usepackage[ruled]{algorithm2e} % For algorithms
\allowdisplaybreaks
\SetAlFnt{\small}
\SetAlCapFnt{\small}
\SetAlCapNameFnt{\small}
\SetAlCapHSkip{0pt}
\IncMargin{-\parindent}

\usepackage[nameinlink,capitalise]{cleveref}

\usepackage{natbib}

\usepackage{comment} 
\usepackage{subcaption}
\usepackage{multirow}
\usepackage{bbold}
\usepackage{pgf,tikz,pgfplots}
\usetikzlibrary{matrix}
\usetikzlibrary{calc,patterns,intersections,pgfplots.fillbetween}
\usepackage[suppress]{color-edits}% suppress
\addauthor{DM}{red}

\newcommand{\dme}[1]{{\DMedit{#1}}}

\addauthor{MF}{blue}

\newcommand{\mfe}[1]{{\MFedit{#1}}}

\addauthor{SM}{purple!80!black}

\newcommand{\sme}[1]{{\SMedit{#1}}}

\addauthor{AE}{orange}

\newcommand{\sig}[1][]{#1{s}}
\newcommand{\val}[1][]{#1{v}}
\newcommand{\alloc}[1][]{#1{x}}
\newcommand{\price}[1][]{#1{p}}
\newcommand{\sigs}[1][]{\mathbf{\sig[#1]}}
\newcommand{\sigspace}{\mathbf{S}}
\newcommand{\vals}[1][]{\mathbf{\val[#1]}}
\newcommand{\allocs}[1][]{\mathbf{\alloc[#1]}}
\newcommand{\prices}[1][]{\mathbf{\price[#1]}}
\newcommand{\low}[3][]{\val[#1]_{#2}(\sigs[#1]_{-#3},0_{#3})}

\newcommand{\remove}[1]{}

\newtheorem{definition}{Definition}[section]
\newtheorem{lemma}{Lemma}[section]
\newtheorem{theorem}{Theorem}[section]

\newtheorem{proposition}{Proposition}[section]

\newcommand{\indep}{\mathcal{I}}
\newcommand{\cands}{\mathcal{C}}

\setcitestyle{acmnumeric}
\title{Private Interdependent Valuations:\\
New Bounds for Single-Item Auctions and Matroids\thanks{The work of A.\ Eden was supported by the Israel Science Foundation (grant No. 533/23). The work of M.\ Feldman, S.\ Mauras, and D.\ Mohan has been partially funded by the European Research Council (ERC) under the European Union's Horizon 2020 research and innovation program (grant agreement No. 866132), by an Amazon Research Award, by the NSF-BSF (grant No. 2020788), and by a grant from TAU Center for AI and Data Science (TAD). This research was partially conducted while some of the authors were visiting the Simons Laufer Mathematical Sciences Institute (formerly MSRI) in Berkeley, California, during the Fall 2023 semester, which was funded by the National Science Foundation (grant No. DMS-1928930) and by the Alfred P. Sloan Foundation (grant G-2021-16778).}}

\author{
		Alon Eden%\\
		\thanks{The Hebrew University; {\tt alon.eden@mail.huji.ac.il}}
		\and
		Michal Feldman%\\
		\thanks{Tel Aviv University; {\tt michal.feldman@cs.tau.ac.il}}
            \and
            Simon Mauras
            \thanks{Tel Aviv University; INRIA; {\tt simon.mauras@inria.fr}}
            \and
            Divyarthi Mohan
            \thanks{Tel Aviv University; {\tt divyarthim@tau.ac.il}}
	}

\begin{document}

% Title page for title and abstract only.
\begin{titlepage}

\maketitle

\begin{abstract}

We study auction design within the widely acclaimed model of interdependent values, introduced by Milgrom and Weber [1982]. In this model, every bidder $i$ has a private signal $s_i$ for the item for sale, and a public valuation function $v_i(s_1,\ldots,s_n)$ which maps every vector of private signals (of all bidders) into a real value. A recent line of work established the existence of approximately-optimal mechanisms within this framework, even in the more challenging scenario where each bidder's valuation function $v_i$ is also private. 
This body of work has primarily focused on single-item auctions with two natural classes of valuations: those exhibiting submodularity over signals (SOS) and $d$-critical valuations.

In this work we advance the state of the art on interdependent values with private valuation functions, with respect to both SOS and $d$-critical valuations.
For SOS valuations, we devise a new mechanism that gives an improved approximation bound of $5$ for single-item auctions.
This mechanism employs a novel variant of an ``eating mechanism'', leveraging LP-duality to achieve feasibility with reduced welfare loss. 
For $d$-critical valuations, we broaden the scope of existing results beyond single-item auctions, introducing a mechanism that gives a $(d+1)$-approximation for any environment with matroid feasibility constraints on the set of agents that can be simultaneously served. Notably, this approximation bound is tight, even with respect to single-item auctions.
\end{abstract}
% Optionally include a table of contents
% \vspace{1cm}
% \setcounter{tocdepth}{2} % adjust to 1 if desired
% \tableofcontents

\end{titlepage}

% Paper body

\section{Introduction} \label{sec:intro}

A standard assumption in auction theory literature is that each bidder has a private value for the item, typically assumed to be independent of other bidders' values. However, in many scenarios, a bidder's value may crucially depend on private information possessed by others. This is the case in contexts such as mineral rights auctions, art auctions, and online ad auctions, where bidders lack knowledge of their value for the item a priori, and moreover, their values are interdependent.
 
The seminal work of \citet{MilgromWeber82} and \citet{wilson1969communications}, recognized by the 2020 Nobel prize in Economics \cite{nobel2021considerations}, introduced the interdependent values model to formally study these intricate settings. In this model, every bidder $i$ has a private signal $s_i$, representing her partial information about the item, and a public valuation function $v_i(\cdot)$ that maps the signals of all bidders to $i$'s value for the item, with $v_i(s_1, \ldots, s_n)$ denoting bidder $i$'s value for the item under a signal profile $(s_1, \ldots, s_n)$. For instance, in the context of art auctions, a bidder's signal might encapsulate her perception of the artwork's significance or her personal connection to it, while her valuation might also be shaped by the information and perspectives of other bidders. 
Naturally, this scenario presents a notably more complex challenge than the standard model of independent private values.
\mfe{In particular,} \dme{a long line of economic research has established strong impossibilities for obtaining optimal welfare truthfully, \mfe{except in settings that satisfy a strict property called \emph{single-crossing}}~\cite{maskin1992,DM00,jm01, ausubel1999generalized}.}

{\bf The approximation lens.}
Recent endeavors in EconCS have addressed these challenges by adopting the algorithmic lens of approximation, leading to the construction of truthful mechanisms that obtain approximately optimal welfare or revenue (e.g.,~\cite{RoughgardenTC16, ChawlaFK14, EdenFFG18, EdenFFGK19,AmerTC21,EdenGZ22,LuSZ22,CohenFMT23,gkatzelis2021prior,ChenEW}).
A key result from these studies is that when valuations satisfy submodularity over signals (SOS) --- capturing valuations with diminishing returns --- it is possible to obtain a constant-factor approximation of the optimal welfare through a truthful auction \cite{EdenFFGK19,AmerTC21,LuSZ22}.
A valuation function satisfies SOS if, for any $j$, the effect of an increase in $s_j$ is more significant when the other signals $s_{-j}$ are lower.
This class of valuations captures a variety of natural settings where information (signals) have decreasing marginal returns, including those frequently explored in the literature, such as
art auctions and mineral rights auctions.
Quite remarkably, these results apply to a much broader range of auction settings than single-item auctions. In particular, \cite{EdenFFGK19} provides a $4$-approximation for welfare for general combinatorial auctions under SOS and separable valuations.

{\bf Private valuation functions.} 
A crucial assumption that drives the above results is that while the signals are private information, the valuation functions mapping signals to values are publicly known.
This assumption may be quite controversial in real-world scenarios, as the way by which a bidder assesses the value of an item, considering all available information from others, may remain private to that bidder.
This concern has led to a new line of work studying the interdependent values model under the setting where both the signals and the valuation functions are {\em private}. This environment proves to be 
significantly more challenging than its public valuations counterpart.
For example, even in single-item auctions with valuations satisfying single-crossing (which enables optimal welfare under public valuations), one cannot guarantee better than the trivial $n$-approximation \cite{EdenGZ22}.

The state of the art results for interdependent valuations with private signals and valuation functions are the following: 
(i) A truthful mechanism for single-item auctions with SOS valuations, that gives a $5.55$-approximation with respect to the optimal welfare \cite{EdenGZ22,EdenFGMM23}.
(ii) A truthful mechanism for single-item auctions with valuations satisfying a natural property called $d$-critical, that gives a  (tight) $(d+1)$-approximation with respect to the optimal welfare~\cite{EdenGZ22}. A valuation $v$ is $d$-critical if, for every signal profile $\sigs=(s_1,\ldots,s_n)$, there exist at most $d$ bidders $j$ for whom reducing $s_j$ changes the value of $v(\sigs)$.

Two primary open problems emerge from the current research landscape.
First, for single-item auctions, there exists a $5.55$-approximation, and it is known that one cannot get a better approximation factor than $2$ \citep{EdenFFGK19}. 
Shrinking this gap is a clear open problem.
Second, work on interdependent settings with private valuation functions has focused primarily on single-item auctions (with some extensions to settings with identical items and unit-demand valuations \cite{EdenFGMM23}). 
This is in stark contrast to interdependent settings with public valuation functions, where the approximation results extend to significantly broader settings. Providing approximation results for private interdependent valuations, in settings beyond single-item auctions, stands as a major open problem.

\subsection{Our Results}

In this paper, we address and make progress on both of the aforementioned open problems.
Our first result is improving the state of the art for single-item actions with SOS valuations: 

\vspace{0.1in}
{\bf Theorem 1} (\Cref{thm:SOS-5approx}): For any single-item auction with private interdependent, SOS valuations, there exists a truthful mechanism that gives $5$-approximation with respect to the optimal social welfare.
\vspace{0.1in}

Beyond improving the approximation factor for SOS valuations, a key component of our result is a new \emph{eating mechanism} for this setting, which might be of independent interest.
The mechanism starts with a single unit of allocation probability. 
Ideally, we envision an eating process wherein each bidder $i$ begins to consume at a moment determined by her value $v_i(\sigs)$ (with those having higher valuations starting first). 
This approach, however, does not yield a truthful mechanism. 
To address this, we implement $n$ separate eating processes, one for each bidder. 
The allocation probability for bidder $i$ then depends on bidder $i$'s true value $v_i(\sigs)$, and on the other bidders' \emph{shadow values} $\low{j}{i}$ that don't depend on $i$'s signal.
The challenge is to construct this mechanism in a way that ensures feasibility, namely that the sum of all allocation probabilities does not exceed $1$, \mfe{while still ensuring good approximation guarantees}. 
To this end, we formulate the eating process solution via a Linear Program (LP) and apply LP duality to bound the allocation probabilities.

Interestingly, the RCF mechanism described in \cite{EdenFGMM23} employs a filtering process that determines a bidder's candidacy based on a random discretization and tie-breaking order. 
In a way, our eating mechanism is analogous to the probabilistic serial rule, whereas the RCF mechanism is analogous to the random serial dictatorship rule.

\vspace{0.1in}

Our second result concerns settings with private interdependent, $d$-critical valuations.
For this valuation class, we extend the prior results beyond single-item auctions, to any auction setting where the set of bidders that can be simultaneously served is subject to matroid feasibility constraints.

\vspace{0.1in}
{\bf Theorem 2} (\Cref{thm:d-critical-matroid}): For any auction setting with matroid feasibility constraints (on the set of bidders that can be simultaneously served) and private interdependent, $d$-critical valuations, there exists a truthful mechanism that gives the tight $(d+1)$-approximation with respect to the optimal social welfare.
\vspace{0.1in}

Interestingly, the approximation ratio of $(d+1)$ is proven to be tight, even for the single-item setting~\cite{EdenGZ22}. In order to extend the result for \emph{any} matroid feasibility constraint, we carefully generalize the mechanism for single-item auction outlined in \cite{EdenGZ22}. 

The high-level intuition for the single-item auction is that a bidder $i$ qualifies as a candidate for allocation if and only if her value $v_i(\sigs)$ exceeds the shadow values $\low{j}{i}$ of all other bidders $j \neq i$. Notably, the bidder $i^\star$ with the true highest value always qualifies as a candidate, with the only additional potential candidates being the $\le d$ bidders whose signals can reduce $v_{i^\star}(\sigs)$. 

Building upon this intuition, we introduce a candidate filtering algorithm that designates bidder $i$ as a candidate if and only if she belongs to the maximal independent set, when using the shadow values of others. A key challenge here is to show that all bidders in the true optimal set $I^\star$ are selected as candidates, and the only other potential candidates are those who can reduce the value of the bidders in $I^\star$.
We then need to address the challenge that the set of candidates may not be feasible. To overcome this issue, we employ the matroid partition theorem of \citet{Edmonds71}, showing that the set of candidates can be partitioned into at most $(d+1)$ independent sets, allowing us to provide an ex-post feasible allocation by selecting each of these independent sets with probability $1/(d+1)$.

We also note that, this mechanism can be adjusted to scenarios where every agent $i$ has a $d_i$-critical valuation, with an unknown $d_i$. In this case, the mechanism can be extended to obtain $2(\max_i(d_i)+1)$-approximation.

\mfe{A natural open problem arising from our work is whether there exists a truthful mechanism for private interdependent, SOS valuations with matroid feasibility constraints.
We believe that our new eating mechanism may be a promising direction for resolving this problem, as it sidesteps the complexities associated with random ordering (as was required in the previous approach taken by \citet{EdenFGMM23}), which is hard to analyze in matroid settings, as evident by the notoriously hard matroid secretary problem \cite{BabaioffIKK18}.}

\subsection{Related Work}
The two most closely related works are \citet{EdenGZ22}, who are the first to study interdependent values with private valuation functions and provide an $O(\log^2 n)$-approximation to the optimal welfare for single-item auctions under SOS valuations, and \citet{EdenFGMM23}, who improve this to a $5.5$-approximation and extend the result to multi-unit auctions. In addition, these papers also establish a $(d+1)$-approximation for \mfe{single-item auctions with} $d$-critical valuations. A more general class of valuations called $d$-self-bounding is considered in \cite{EdenFGMM23}, for which  an $O(d)$-approximation is provided. \citet{DM00} study a related setting where the valuation function is unknown to the seller but the bidders \mfe{are aware of each others'} valuation functions. They design a mechanism where the bidders report a complicated \dme{contingent} bidding function, and show that under single-crossing type conditions there is a fully efficient equilibrium.

The interdependent model with public valuations has attracted extensive work in economics for decades, establishing strong impossibilities for obtaining optimal welfare truthfully unless the valuations satisfy some strong condition such as single-crossing \cite{maskin1992,DM00,jm01, ausubel1999generalized}. 
\mfe{Many additional studies have explored the interdependent values model in a variety of scenarios, often incorporating single-crossing type assumptions (e.g.,~\cite{ChungE02,ItoP06,RobuPIJ13,CKK15}).}

\mfe{Recently, there has been a surge of research within EconCS focusing on the interdependent values model from an algorithmic perspective.} \citet{RoughgardenTC16} study prior-independent mechanisms that are approximately optimal in the revenue maximization regime for single parameter settings with downward closed constraints and further provide conditions under which Myerson-like virtual welfare maximization obtains optimal revenue. \citet{Li13} provides simple VCG-based mechanisms that obtain near-optimal revenue under MHR distributions and matroid feasibility constraints. \citet{ChawlaFK14} provide approximation guarantees for revenue in single-parameter settings with matroid feasibility constraints under relaxed assumptions. These results assume some form of single-crossing type condition in order to obtain positive results. 

\citet{EdenFFG18} provide approximation guarantees for welfare in prior-free single-item auction settings when the valuation functions are \emph{approximately} single-crossing. \citet{EdenFFGK19} were the first to study public SOS valuations in interdependent setting and provide a $4$-approximation for any single-parameter setting with downward closed feasibility constraints (without any single-crossing type assumption). They also study combinatorial auctions and provide a $4$-approximation under an additional condition of seperable valuations. \citet{AmerTC21} and \citet{LuSZ22} provide improved approximation guarantees for the single-item setting. \citet{EdenFTZ21} study price of anarchy bounds for simple non-truthful mechanisms, and \citet{gkatzelis2021prior} show that clock auctions can obtain approximately optimal welfare under certain assumptions on valuation functions and finite signal space. \citet{CohenFMT23} study approximation guarantees for the public projects setting with interdependent values, and \citet{BirmpasELR23} study a fair division problem with interdependence.

\section{Preliminaries} \label{sec:prelims}

\subsection{Interdependent Valuations and Truthful Mechanisms} \label{sec:prelim-model}

We consider an auction with $n$ bidders with interdependent valuations. Every bidder $i \in [n]$ has a private signal $\sig_i \in S_i$, capturing bidder $i$'s private information about the item, where $S_i$ denotes the signal space of bidder $i$. 
Wlog assume that the minimum signal in $S_i$ is $0$, for every $i$.
We denote by $\sigspace = S_1\times \ldots \times S_n$ the joint signal space of the bidders, and by $\sigs = (\sig_1, \ldots, \sig_n) \in \sigspace$ a signal profile.
As is standard, we denote by $\sigs_{-i}=(\sig_1, \ldots, \sig_{i-1},\sig_{i+1},\ldots, \sig_n)$ the signal profile of all bidders other than bidder $i$.

In addition \mfe{to the signal}, every bidder $i$ has a private \dme{monotone} valuation function $\val_i: \sigspace\rightarrow \mathbb R_+$, which maps \mfe{every signal profile $\sigs = (\sig_1, \ldots, \sig_n)$} into a non-negative real number, which is bidder $i$'s value for the item.
We denote by $V_i \subseteq \mathbb R_+^{\sigspace}$ the valuation space of bidder $i$, and by $\mathbf{V} = V_1\times\ldots\times V_n$ the joint valuation space of all bidders. 
A vector $\vals = (\val_1, \ldots, \val_n) \in \mathbf{V}$ denotes a valuation profile. 

We refer to bidder $i$'s value under signal profile $\sigs$ as $i$'s real value. 
We sometimes refer to a low estimate of bidder $i$'s value, where bidder $j$'s signal is zeroed out; i.e., $v_i(\sigs_{-j},0_j)$; we refer to this value as bidder $i$'s shadow value.

\mfe{A mechanism is defined by a pair $(\allocs, \prices)$ denoting the mechanism's allocation and payment rules, based on the bidders' reports of their signals and valuations. 
The allocation rule $\allocs: \sigspace\times\mathbf{V}\rightarrow [0,1]^n$ returns the allocation probability $\alloc_i(\sigs,\vals)$ for every bidder $i$, and the payment rule $\prices: \sigspace\times\mathbf{V}\rightarrow 
\mathbb{R}_+^n$ returns the payment $\price_i(\sigs,\vals)$ every bidder $i$ for reported signals $\sigs,\vals$.}

Unless specified otherwise, we access bidder valuations via value queries; namely, given a signal profile $\sigs$, bidder $i$'s value oracle $\val_i$ returns $\val_i(\sigs)$. 
A mechanism is said to be polynomial if it makes a polynomial number of value queries. 

A mechanism $(\allocs,\prices)$ is said to be {\em truthful} if it is an ex-post Nash equilibrium for the bidders to truthfully report their private information (signals and valuations).
In our query access model, truthfulness means that it is in every bidder's best interest to answer every query truthfully, given that other bidders do the same.

\begin{definition}[EPIC-IR]
A mechanism $(\mathbf{x}, \mathbf{p})$ is 
{\em ex-post incentive compatible (IC)} if for every $i\in[n],\sigs\in \sigspace,  \vals \in \mathbf{V},  \sig_i'\in S_i, \val_i'\in V_i$
\begin{equation}
x_{i}(\sigs, \vals) \cdot v_{i}(\sigs) - p_i(\sigs, \vals) \geq x_{i}(\sigs_{-i},\sig_i', \vals_{-i}, \val'_i) \cdot v_{i}(\sigs) - p_i(\sigs_{-i},\sig_i', \vals_{-i}, \val'_i).
\label{eq:IC}
\end{equation}
It is {\em ex-post individually rational (IR)} if for every $i\in[n]$, $\sigs\in\sigspace$, and $\vals\in\mathbf{V}$
\begin{equation}
x_{i}(\sigs, \vals) \cdot v_{i}(\sigs) - p_i(\sigs, \vals) \geq 0
\label{eq:IR}
\end{equation}
It is EPIC-IR if it is both ex-post IC and ex-post IR.
An allocation $\allocs$ is {\em EPIC-IR} implementable if there exists a payment rule $\prices$ such that the pair $(\allocs,\prices)$ is EPIC-IR. 
\end{definition}

{It is well known that even when the valuation functions are public, this is the strongest possible solution concept when dealing with interdependent valuations.\footnote{Dominant strategy incentive-compatibility does not make sense, as a bidder $i$ might not be willing to win if other bidders over-bid, which causes the winner to over-pay and incur a negative utility.}}

\citet{EdenGZ22} give a sufficient condition for an allocation rule $\allocs$ to be EPIC-IR implementable.

\begin{proposition}[\citet{EdenGZ22}]\label{lem:pricing}
An allocation rule $\allocs$ is EPIC-IR implementable if for every bidder $i$, $\alloc_{i}$ depends only on $\sigs_{-i}, \vals_{-i}$ and $\val_{i}(\sigs)$, and is non-decreasing in $\val_{i}(\sigs)$. 

For an (EPIC-IR) implementable $\allocs$, the corresponding payment rule $\prices$ is given by:
\begin{equation}
p_{i}(\sigs, \vals) := x_i(\sigs_{-i}, \vals_{-i}, \val_i(\sigs))\cdot \val_i(\sigs) - \int_0^{\val_i(\sigs)} x_i(\sigs_{-i}, \vals_{-i}, t) \, \mathrm dt.
\label{eq:pricing}
\end{equation}
\end{proposition}
That is, bidder $i$'s allocation may depend on all other bidders' signals and valuation functions, and it can only depend on bidder $i$'s signal $\sig_i$ or valuation function $v_{i}$ through the numerical value $v_i(\sigs)$. {\citet{EdenGZ22} show that this condition is almost necessary in order to be EPIC-IR implementable.\footnote{{The necessary conditions for EPIC-IR implementablity are (i) $x_i$ is monotone in $v_i(\sigs)$, and (ii) for a given $\sigs_{-i}$, the set of signals $s_i,s'_i$ and valuation functions $v_i,v'_i$ such that $v_i(s_i,\sigs_{-i}) = v'_i(s'_i,\sigs_{-i})$ and $x_i(v_i,s_i,\sigs_{-i}) \neq x_i(v'_i,s'_i,\sigs_{-i})$ has measure $0$.}}}

\subsection{Valuations Classes}
We focus on two  classes of valuations which are well-studied in the context of interdependent values with private valuation functions, \mfe{namely SOS valuations and $d$-critical valuations}.

\begin{definition}[SOS Valuations]
    A valuations function $v:\sigspace\rightarrow \mathbb R_+$ is Submodular over signals (or SOS) if for  every $i$, $s_i\in S_i$, $\sigs_{-i}, \hat{\sigs}_{-i}\in \sigspace_{-i}$ such that $\sigs_{-i}\preceq \hat{\sigs}_{-i}$, and $\delta>0$,
    \begin{eqnarray}
        v(s_i+\delta,\sigs_{-i})- v(s_i,\sigs_{-i}) \ge v(s_i+\delta,\hat{\sigs}_{-i})- v(s_i,\hat{\sigs}_{-i}).
    \end{eqnarray}
\end{definition}

The SOS class includes most valuation classes studied in the literature on interdependent values, including the resale model (i.e., affine functions)~\dme{\cite{myerson1981optimal,Klemperer98}}, the mineral rights model~\dme{\cite{wilson1969communications}}, average value, etc.

Previous work \cite{EdenGZ22,LuSZ22,EdenFGMM23} has shown that SOS valuations satisfy a useful property called self-bounding, where $\sum_{i=1}^n (\val(\sigs) - \low{}{i}) \leq \val(\sigs)$.
This is cast in the following lemma.

\begin{lemma}[SOS functions are self-bounding
\cite{EdenGZ22,LuSZ22}]
\label{lem:self-bounding}
For every monotone SOS valuation function $v: \sigspace\rightarrow \mathbb R_+$, and every signal profile $\sigs$, we have
$$
\sum_{i=1}^n (\val(\sigs) - \low{}{i}) \leq \val(\sigs).
$$
\end{lemma}

We now turn to define $d$-critical valuations.

\begin{definition}[$d$-critical Valuations]
    A valuation function $v:\sigspace\rightarrow \mathbb R_+$ is $d$-critical if for every $\sigs\in\sigspace,$ 
    \begin{eqnarray}
        |\{j\ :\ \low{}{j} < v(\sigs)\}| \le d.    
    \end{eqnarray}
\end{definition}

A few natural functions that lie in low levels of the $d$-critical hierarchy are: (i) max signal is $1$-critical, (ii) any weighted matroid function over signals with rank $d$ is $d$-critical,
\sme{(iii)} values derived from a social network with $d$-bounded degrees (where every bidder's value is an arbitrary function of its neighbors' signals) \sme{are $d$-critical}. 

\citet{EdenGZ22} show the following impossibility for (a special case of) $d$-critical valuations.

\begin{proposition}[\citet{EdenGZ22}] \label{prop:d_critical_lb}
    For every EPIC-IR single-item auction for $d$-critical valuations, there exists an instance for which the obtained welfare cannot give better than $(d+1)$-approximation for the optimal welfare.
\end{proposition}

\subsection{Matroid Preliminaries} \label{sec:matroid_prelims}

In Section~\ref{sec:matroids}, we consider auctions where the feasible sets of bidders that can be served are independent sets of  a matroid. \mfe{The definition of a matroid follows.}

\begin{definition}[Matroid]
A matroid is a pair $M=([n],\indep)$, where $[n]$ is a set of elements, and $\indep$ is a (non-empty) collection of subsets of $[n]$, which are the independent sets of the matroid, satisfying: 
\begin{itemize}
    \item Downward closed: If $I \in \indep$ then $A\in \indep$ for all subsets $A\subseteq I$.
    \item Exchange property: If $A,B \in \indep$ with $|A| < |B|$ then there exists $i\in B\setminus A$ such that $A\cup\{i\} \in \indep$.
\end{itemize}
\end{definition}

\begin{definition}[Rank Function]
    For a subset of elements $S\subseteq [n]$, let $rank(S)=\max_{I\subseteq S\ : \ I\in \indep} |I|$. $rank(M)=rank([n])$ is referred to as the matroid's \textit{rank}.

    The rank function is submodular \cite{Oxley2011}. That is, for every $S\subseteq T\subseteq [n]$, and for every $i\in [n]$, $rank(S\cup \{i\})-rank(S)\ge rank(T\cup \{i\})-rank(T).$
\end{definition}

\mfe{{\bf Greedy algorithm.} Consider a matroid with weighted elements, where the weight function $w:[n]\rightarrow\mathbb{R}^+$ assigns a non-negative weight to every element $i \in [n]$. The greedy algorithm with weights $w$ operates by sorting the elements in $[n]$ in non-increasing order of $w$ and then sequentially adds elements to maintain an independent set. Specifically, the algorithm starts with an empty set $I^\star$ and, for each element $i$ in the specified order, adds $i$ to $I^\star$ if $I^\star\cup\{i\} \in \indep$.
It is well known that the greedy algorithm produces an independent set $I^\star$ with maximum weight. This is cast in the following theorem.}

\begin{theorem}[Greedy is Optimal (\citet{Edmonds71})] \label{thm:greedy}
    % Consider $I^\star$ which results from the running greedy algorithm on matroid $M=([n],\indep)$, then
    \mfe{Let $I^\star$ be the set returned by the greedy algorithm on a matroid $M=([n],\indep)$. It holds that}
    \begin{enumerate}
        \item $|I^\star|=rank(M)$.
        \item $\sum_{i\in I^\star}w(i) = \max_{I\in \indep}\sum_{i\in I} w(i)$.
    \end{enumerate}
\end{theorem}

For our result, we need the following lemma about the greedy algorithm.
\begin{lemma} \label{lemma:before_still_selected}
    Let $i\in [n]$ be some element in $M$. Consider two weight functions  $w,\hat{w}$, where $\hat{w}(i)\ge w(i)$ and $\hat{w}(j)\le w(j)$ for every $j\neq i$. If $i$ is selected by the greedy algorithm with weights $w$ and some tie-breaking rule, then $i$ is also selected by the greedy algorithm with weights $\hat{w}$ and the same tie-breaking rule.
\end{lemma}

\begin{proof}
    Consider the ordering $\sigma$ which results by ordering the elements according to $w$ using some tie-breaking rule. By renaming, assume we break ties in favor of lower-index elements. Let $I^\star$ be the obtained independent set. Consider an element $i\in I^\star$ and consider $\hat{w}$ such that $\hat{w}(i)\ge w(i)$ and $\hat{w}(j)\le w(j)$ for $j\neq i$. Let $\hat{\sigma}$ and $\hat{I}^\star$ be the ordering and the independent set determined according to $\hat{w}$ when breaking ties in favor of lower-index elements.  Let $A$ and $\hat{A}$ be the sets of elements that appear before $i$ when ordering the elements according to orders $\sigma$ (i.e., according to $w$) and $\hat{\sigma}$ (i.e., according to $\hat{w}$), respectively. Since in order $\hat{\sigma}$, $\hat{w}(i)\ge w(i)$ while $\hat{w}(j)\le w(j)$ for $j\neq i$,  $\hat{A}\subseteq A$.  We transform $\sigma$ into $\hat{\sigma}$ in three steps, and show that after each step, $i$ is still chosen by the greedy algorithm (see Figure~\ref{fig:stages-tikz} for a schematic depiction of the transitions).
    
        First, consider an ordering $\sigma'$ which reorders elements in $A$ to have  the same relative ordering as in $\hat{\sigma}$, but does not change the  ordering of all other elements. In $\sigma'$, the set of elements ordered before $i$ is still $A$. Moreover, the first $|\hat{A}|$ elements in order $\sigma'$ are exactly $\hat{A}$. Let $I^\star_i$ and $I'_i$ be the sets of elements added by running the greedy algorithm on orders $\sigma$ and $\sigma'$, respectively, just before element $i$'s turn. By \Cref{thm:greedy}, $|I^\star_i|=|I'_i|=rank(A)$. Since $i$ is added to $I^\star$ in ordering $\sigma$, $I^\star_i\cup\{i\}\in \indep\rightarrow rank(A\cup \{i\})= rank(A)+1=|I'_i|+1$. Again, by \Cref{thm:greedy}, it must be that $i$ is added to $I'_i$ in ordering $\sigma'$. 

    For the second step, we move $i$ to be right after set $\hat{A}$ in the ordering $\sigma'$, bumping up one spot each element in $A\setminus \hat{A}$. Let $\sigma''$ be the resulting ordering. Notice that up until the position of element $i$ in the order, \textit{including}, the elements are ordered the same as in $\hat{\sigma}$. Let $I''_i$ be the set of elements added by the greedy algorithm in order $\sigma''$ up until element $i$'s turn. Since up until this point, the order of the elements in $\sigma'$ and $\sigma''$ is the same, $I''_i\subseteq I'_i$. Therefore, since $\indep$ is downwards-closed, it must be the case that $i$ can be added to $I''_i$ as well.

    To finish the proof, we now order the elements after $i$ to be in the same order as in $\hat{\sigma}$, thus recovering order $\hat{\sigma}$. Notice that elements ordered after $i$ in an instance of the greedy algorithm cannot determine whether $i$ is selected or not. Therefore, $i$ is selected by running greedy on $\hat{\sigma}$ to be in $\hat{I}^\star$, finishing the proof.
\end{proof}

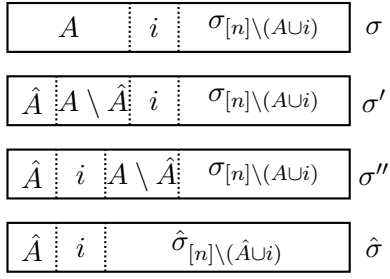
\begin{figure}[ht]
    \centering
\begin{tikzpicture}[scale = 0.65,thick]
    % last rectangle
    \draw (0,0) rectangle (7,1);
    \node at (0.5,0.5) {$\hat A$};
    \node at (1.5,0.5) {$i$};
    \node at (4.5,0.5) {$\hat{\sigma}_{[n] \setminus (\hat A\cup i)}$};
    \node at (7.5,0.5) {$\hat \sigma$};
    \foreach \x in {1,2}
        \draw[densely dotted] (\x,0) -- (\x,1);
    
    % Third rectangle
    \draw (0,1.5) rectangle (7,2.5);
    \node at (0.5,2) {$\hat A$};
    \node at (1.5,2) {$i$};
    \node at (2.75,2) {${A} \setminus \hat{A} $};
    \node at (5.25,2) {$\sigma_{[n] \setminus (A\cup i)}$};
    \node at (7.5,2) {$\sigma''$};
    \foreach \x in {1,2,3.5}
        \draw[densely dotted] (\x,1.5) -- (\x,2.5);
    
    % Second rectangle
    \draw (0,3) rectangle (7,4);
    \node at (0.5,3.5) {$\hat A$};
    \node at (1.75,3.5) {${A} \setminus \hat{A} $};
    \node at (3,3.5) {$i$};
    \node at (5.25,3.5) {$\sigma_{[n] \setminus (A\cup i)}$};
    \node at (7.5,3.5) {$\sigma'$};
    \foreach \x in {1,2.5,3.5}
        \draw[densely dotted] (\x,3) -- (\x,4);
    
    % First rectangle
    \draw (0,4.5) rectangle (7,5.5);
    \node at (1.25,5) {$A$};
    \node at (3,5) {$i$};
    \node at (5.25,5) {$\sigma_{[n] \setminus (A\cup i)}$};
    \node at (7.5,5) {$\sigma$};
    \foreach \x in {2.5,3.5}
        \draw[densely dotted] (\x,4.5) -- (\x,5.5);
\end{tikzpicture}
\caption{The transitions from $\sigma$ to $\hat{\sigma}$ used in the proof of \Cref{lemma:before_still_selected}.}
\label{fig:stages-tikz}
\end{figure}

We will also use the following known theorem about matroids.

\begin{theorem}[Matroid partition~(\citet{edmonds1965minimum})] \label{thm:partition}
Given a matroid $M$ over a ground set $[n]$, it is possible to find a partitioning $M = I_1 \uplus \dots \uplus I_t$ with $t$ pairwise disjoint independent sets $I_1, \dots, I_t \in \indep$ if and only if
\begin{eqnarray}
\forall S \subseteq [n],\quad  |S| \leq  t\cdot rank(S). \label{eq:partition_condition}
\end{eqnarray}

%where $r(S)$ is the rank of $S$ in $M$. 
Moreover, there exist efficient algorithms to find such partitioning.
\end{theorem}

%%%%%%%%%%%%%%%%%%%%%%%%%%%%%%%%%%%%%%%%
\section{Eating Mechanism for Single-Item Auctions with SOS Valuations} \label{sec:eating}

\mfe{In this section, we present a truthful mechanism that improves upon the state-of-the-art approximation for single-item auctions with SOS valuations.} 
\dme{We start by introducing a new \emph{eating mechanism}, and then show that this mechanism is truthful and provides a $5$-approximation to the optimal welfare for SOS valuations. The main theorem of this section follows.}

\begin{theorem}\label{thm:SOS-5approx}
    \mfe{Mechanism~\ref{alg:eating} (the eating mechanism) is}
    %There exists 
    an EPIC-IR mechanism that obtains a $5$-approximation to the optimal welfare when the bidders have interdependent values with private SOS valuations.
\end{theorem}

\mfe{The proof of Theorem~\ref{thm:SOS-5approx} unfolds as follows. In Lemma~\ref{lem:eating-truthful} we show that the mechanism is EPIC-IR, in Lemma~\ref{lem:5-approx} we show that the mechanism, if feasible, gives $5$-approximation to the optimal welfare, and in Lemma~\ref{lem:eating_feasible} we prove that the mechanism is feasible. Collectively, these lemmas prove Theorem~\ref{thm:SOS-5approx}.}

\mfe{Before presenting the lemmas, we give an informal description of the mechanism, and present a linear program whose solution coincides with the allocation probabilities produced by the eating mechanism (as shown in Lemma~\ref{lem:eating-lp}).}

\mfe{{\bf The eating mechanism.} Our mechanism considers $n$ separate eating processes, one for each bidder, where the eating process corresponding to bidder $i$ determines $i$'s share of the allocation probability. 
In bidder $i$'s eating process, we use $i$'s real value, and for all other bidders $j \neq i$, we use their shadow values (i.e., their values when $i$'s signal is zeroed out). 
In particular, every bidder $j$ starts eating at a time that depends on her (shadow) value (with higher valued bidders starting earlier). 
At each point in time, all the bidders eating at that time, eat at the same speed. The eating process ends when the total share eaten is $1$.\footnote{Note that some bidders may not even start eating before the process ends, leading to a share of 0.}  Bidder $i$ is then allocated her (normalized) share obtained in this process. \Cref{fig:eating} depicts the eating process used to determine $i$'s allocation.}

\begin{figure}[ht]
    \centering
    %\includegraphics{}
    %\vspace{0.5cm}
    \begin{tikzpicture}
    % Draw the horizontal line
    \draw[thick,->] (0,0) -- (10,0);

    \draw[thick] (0,0.1) -- (0,-0.1);

    % t at right end
    \node[anchor=west] at (10,0) {time};
    
    % Draw additional small circles
    \filldraw (2,0) circle (1.5pt) node[below] {$-\ln \low{j}{i}$};
    \filldraw (4,0) circle (1.5pt) node[below] {$-\ln v_i(\sigs)$};
    \filldraw (8,0) circle (1.5pt) node[below] {$-\ln \low{j'}{i}$};

     % Draw a small square at (6,0)
        \filldraw[red] (6-0.05,-0.05) rectangle +(3pt,3pt); \node[below,anchor=north] at (6,0) {$t$};

    % Draw a small line segment below the main line (red)
    \draw[thick, dashed, blue] (2,0.3) -- (6,0.3);
    \draw[thick, blue] (2,0.4) -- (2,0.2);
    \draw[thick,blue] (6,0.4) -- (6,0.2);
    
    % Draw a small line segment below the red line (blue)
    \draw[thick, blue] (4,0.6) -- (6,0.6);
    \draw[thick, blue] (4,0.7) -- (4,0.5);
    \draw[thick, blue] (6,0.7) -- (6,0.5);

    % Draw vertical dotted lines
    \draw[densely dotted] (2,0) -- (2,0.3);
    \draw[densely dotted] (6,0.2) -- (6,0);
    
    \draw[densely dotted] (4,0) -- (4,0.6);
    \draw[densely dotted] (6,0.6) -- (6,0);
\end{tikzpicture}
    \caption{The eating process used to determine bidder $i$'s allocation. bidder $i$ starts eating at time $-\ln v_i(\sigs)$ and all other bidders $j\neq i$ start eating at time $-\ln \low{j}{i}$. At each point in time, all bidders eating at that time eat at the same speed. The solid blue line denotes bidder $i$'s share and the dashed blue line denotes $j$'s (pretend) share. The red square denotes the time $t$ when the sum of the blue lines adds up to $1$, thus halting the eating process.
    }
    \label{fig:eating}
\end{figure}
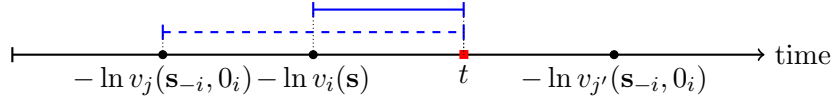

\begin{algorithm}[h]
%\textbf{Function} $\texttt{eat}(w_i, \{w_j\}_{j\neq i}):$
\textbf{Function} $\texttt{eat}(\text{weight function }w:[n]\rightarrow \mathbb R_+):$
\begin{enumerate}
%\item Each bidder $j$ starts eating at time $-\ln(w_j)$, and eats at a constant speed of $1$.
\item Each bidder $j$ starts eating at time $-\ln(w(j))$, and eats at a constant speed of $1$.
\item Eating stops at time $t$, when the item has been entirely eaten, that is when $\sum_{j=0}^n (t+\ln w(j))^+ = 1$.
\item Return the allocation, that is the vector $\textbf y$ where $y_j = (t+\ln w(j))^+$.
\end{enumerate}
\textbf{Mechanism} $\texttt{eating}:$
\begin{enumerate}
\item Elicit signals $\sigs[\hat]$ and valuation functions $\vals[\hat]$.
\item For each bidder $i$:\\
\quad Define the function $w_i$ where $w_i(i) = \val[\hat]_i(\sigs[\hat])$ and $w_i(j)=\low[\hat]{j}{i}$ for $j \neq i$.\\
\quad Set $x_i \gets \texttt{eat}(w_i)_i / 4$, that is, the $i$-th coordinate of $\texttt{eat}(w_i)$, divided by $4$.
\item Charge payments using Eq~\eqref{eq:pricing}.
\end{enumerate}
\renewcommand{\algorithmcfname}{Mechanism}
\caption{\textbf{Eating Mechanism}.}
\label{alg:eating}
\end{algorithm}

\sme
{For the sake of the analysis, we formulate the eating process as the solution of the following linear program ($P_i$):
\begin{align*}
\textbf{maximize}\quad& y_i&(P_i)\\
\text{such that}\quad & y_i \leq y_j+\ln(w(i))-\ln(w(j)) \quad  \forall j\neq i\\
& \textstyle\sum_i y_i \leq 1\\
&y_1, \dots, y_n \geq 0.\\
\end{align*}
\Cref{lem:eating-lp} shows that the (pre-normalized) share allocated to bidder $i$ by the eating mechanism equals the value of the  linear program $(P_i)$.

To establish the feasibility of Mechanism~\ref{alg:eating}, we will show in \Cref{lem:eating_feasible} that the sum of probabilities assigned to the bidders does not exceed 1. We obtain an upper-bound on the (pre-normalized) share of bidder $i$ in the primal program $(P_i)$ using a feasible solution to the dual program $(D_i)$ presented below.

\begin{align*}
\textbf{minimize}\quad& \alpha_i+\textstyle\sum_{j\neq i} \beta_{i,j} \cdot (\ln(w(i))-\ln(w(j)))&(D_i)\\
\text{such that}\quad& \alpha_i + \textstyle\sum_{j\neq i}\beta_{i,j} \geq 1\\
 &\alpha_i \geq \beta_{i,j} \quad \forall j\neq i \\
& \alpha_i \geq 0, \quad \;\beta_{i,j} \geq 0 \quad \forall j\neq i.
\end{align*}}

\mfe{The following lemma shows that the value of the primal LP $(P_i)$ of bidder $i$ equals the (pre-normalized) share of bidder $i$ in the eating mechanism.}

\begin{lemma}\label{lem:eating-lp}
Given a weight function $w$, the function \texttt{eat}$(w)$ assigns bidder $i$ a share that equals the value of the linear program $(P_i)$ if it is feasible, and $0$ otherwise.
In particular, it is non-decreasing in $w(i)$, and \sme{non-increasing} in every $w(j)$ with $j\neq i$.
\end{lemma}
\begin{proof}
\sme{First, we give another linear program $(P_i')$ which is analogous to the eating procedure. We show that $(P_i)$ is feasible if and only if $(P_i')$ has a non-negative solution, in which case they both have the same value.
}

\begin{align*}
\textbf{maximize}\quad& t+\ln(w(i))&(P_i')\\
\text{such that}\quad& t+\ln(w(j)) \leq z_j \qquad \forall j\in[n]\quad &\\
&\textstyle\sum_i z_i \leq 1\\
& z_1, \dots, z_n \geq 0.
\end{align*}

\sme{Observe that }\dme{$t$ denotes the stopping time of the eating process, given that each bidder $j$'s share is at least $t + \ln (w_j)$ and the total share sums up to at most $1$.} Since $t$ can be negative, $(P_i')$ is always feasible. By construction, the $i$-th coordinate of $\texttt{eat}(w)$ is equal to $(t+\ln(w(i)))^+$. When $t+\ln(w(i)) \geq 0$, setting $y_j=z_j$ and $y_i=t+\ln(w(i))$ gives a feasible solution of the linear program $(P_i)$. Conversely, given a feasible solution $(y_1, \dots, y_n)$ to $(P_i)$, setting $z_j = y_j$ and $t = y_i-\ln(w(i))$ gives a solution to $(P_i')$  with value $t+\ln(w(i)) \geq 0$.

To show the monotonicity properties, observe that increasing $w(i)$ or decreasing $w(j)$ makes the constraints of $(P_i)$ looser, which cannot decrease the objective or make the program infeasible.
\end{proof}

\mfe{We are now ready to prove Theorem~\ref{thm:SOS-5approx}. We first show that Mechanism \ref{alg:eating} is truthful.}

\begin{lemma}
\label{lem:eating-truthful}
Mechanism~\ref{alg:eating} is EPIC-IR.
\end{lemma}

\begin{proof}
By \Cref{lem:pricing}, it suffices to show that $i$'s allocation depends only on $\sigs[\hat]_{-i}$, $\vals[\hat]_{-i}$, and $\val[\hat]_i(\sigs[\hat])$, and is monotone non-decreasing in $\val[\hat]_i(\sigs[\hat])$. 
The first property holds by design, and monotonicity holds by \Cref{lem:eating-lp}, showing that $x_i$ is increasing in $\val[\hat]_i(\sigs)$.
\end{proof}

\mfe{Given that the mechanism is truthful, we hereafter simplify notation and write $\sigs$ and $\vals$ instead of $\sigs[\hat]$ and $\vals[\hat]$.
The following lemma establishes the approximation factor of the mechanism.}

\begin{lemma}
\label{lem:5-approx}
Mechanism \ref{alg:eating} gives $5$-approximation to the optimal welfare.
\end{lemma}

\begin{proof}
We define the weight function $w^\star(j) = \val_j(\sigs)$ mapping each bidder to their true value. First, by monotonicity of the function \texttt{eat}, we have that each bidder $i$ gets an allocation probability $x_i = \texttt{eat}(w_i)_i/4$ which is at least
$\texttt{eat}(w^\star)_i/4$, because we replaced low estimates with the true value for each bidder $j\neq i$ \dme{(i.e., $w^\star(j)\ge w_i(j)$ for all $j$)}. Hence, the expected welfare is at least
$$
\sum_{i=1}^n \val_i(\sigs) \cdot \frac{\texttt{eat}(w^\star)_i}{4}.
$$
Let $i^\star$ be the bidder with the highest value, and let $y_{i^\star}$ be the \dme{(pre-normalized) probability share} she receives from the eating procedure \dme{with weight function $w_{i^\star}$}. Observe that\dme{, in this eating process, all} bidders stop eating \dme{by} time $y_{i^\star}-\ln(\val_{i^\star}(\sigs))$, and thus the \dme{$(1-y_{i^\star})$} \dme{share} which has not been eaten by $i^\star$ has been eaten by someone who has value at least $\exp(\ln(\val_{i^\star}(\sigs))-y_{i^\star})$. Therefore, the overall welfare is greater than or equal to
$$\val_{i^\star}(\sigs) \cdot  \frac{y_{i^\star}+(1-y_{i^\star}) e^{-y_{i^\star}}}{4}.$$
The function $y+(1-y)e^{-y}$ has a minimum approximately equal to $0.8005$ at $y\approx 0.44$. Therefore, after \dme{normalizing the probability shares} by $4$, the expected welfare is at least $\val_{i^\star}(\sigs)\cdot \frac{0.8005}{4} \geq \val_{i^\star}(\sigs)/5$.
\end{proof}

\mfe{Finally, the following lemma confirms the mechanism's feasibility.}

\begin{lemma}\label{lem:eating_feasible}
Mechanism~\ref{alg:eating} is feasible when bidders have SOS valuation functions; that is, the sum of probabilities $x_1+\dots+x_n$ is at most $1$.
\end{lemma}

\begin{proof}
First, define the set $S$ of bidders $i$ who receive a positive probability $x_i > 0$ from the algorithm. For the sake of the analysis, observe that we can restrict our instance to bidders in $S$, because doing so cannot decrease the sum of probabilities. Indeed, removing $j\notin S$ does not decrease their probability $x_j=0$, and when computing the probability of a bidder $i\in S$, removing $j$ is equivalent to replacing $\low{j}{i}$ by $0$, which cannot decrease $x_i$ by the monotonicity property of \Cref{lem:eating-lp}. For simplicity of notations, we assume (without loss of generality) that $S = [n]$, i.e. we relabel bidders and set $n$ to be the size of $S$.

We are now ready to start bounding the probabilities, and we define
$$
\forall i\in [n],\quad
y_i = \texttt{eat}(w_i)_i > 0,
\qquad \text{where }w_i(i) = \val_i(\sigs)\text{ and }w_i(j) = \low{j}{i}\text{ for }j\neq i
$$
and our goal is to give an upper show that $y_1 + \dots + y_n \leq 4$. For all $i$, we consider the dual $(D_i)$ of the linear program $(P_i)$ from \cref{lem:eating-lp}. We now define a feasible solution to all dual linear programs. The main insight is that we will choose symmetric $\beta_{i,j}:=\gamma_{i,j}\cdot \gamma_{j,i}/n$, where $\gamma_{i,j}$'s will be set later. This way, we will be able to swap the indices $i$ and $j$ in $\low{i}{j}$ which appear in the upper-bound:
\begin{align*}
\sum_{i=1}^n y_i &\leq\sum_{i=1}^n \alpha_i + \sum_{i=1}^n \sum_{j\neq i} \beta_{i,j} \cdot (\ln(\val_i(\sigs))-\ln(\low{j}{i}))\\
&=\underbrace{\sum_{i=1}^n \alpha_i}_{A} + \underbrace{\sum_{i=1}^n \sum_{j\neq i} \beta_{i,j} \cdot (\ln(\val_i(\sigs))-\ln(\low{i}{j}))}_{B}.
\end{align*}
Next, using \Cref{lem:self-bounding}, we know that $\sum_{j\neq i}(1-\low{i}{j}/\val_i(\sigs)) \leq 1$. Thus, we want to choose $\gamma_{i,j}$ appropriately to be able to bound the sum $B$:
$$
\forall i\neq j,\qquad\gamma_{i,j} := \begin{cases}
\frac{1-\low{i}{j}/\val_i(\sigs)}{\ln(\val_i(\sigs)) - \ln(\low{i}{j})} &\text{if }\low{i}{j} < \val_{i}(\sigs)\\
1 &\text{if }\low{i}{j} = \val_{i}(\sigs)
\end{cases}.
$$
Using the fact that $-x \ln x \leq 1-x$ for all $0 < x < 1$, we have the following important property:
$$
\forall i\neq j,\qquad
\low{i}{j}/\val_i(\sigs) \leq \gamma_{i,j} \leq 1.
$$
We set $\alpha_i := 1-\sum_{j\neq i} \beta_{i,j}$. By construction, the inequality $\alpha_i+\sum_{j\neq i} \beta_{i,j} \geq 1$ is satisfied. Moreover, we have that each $\beta_{i,j}$ is smaller than $1/n$, therefore each $\alpha_i$ is at least $1/n$ and the inequalities $\alpha_i \geq \beta_{i,j}$ are also satisfied.

To conclude the proof, it remains to give an upper bound on $A$ and $B$.
\begin{align*}
A &= \sum_{i=1}^n \Big(1-\sum_{j\neq i} \beta_{i,j}\Big) = 1+\sum_{i=1}^n\sum_{j\neq i} \Big(\frac{1}{n}-\beta_{i,j}\Big)
&(\dme{\text{using } \alpha_i=1-\sum_{j\neq i}\beta_{i,j}})
%(\text{replacing $\alpha_i$ with its expression})
\\
&= 1+\frac{1}{n}\sum_{i=1}^n\sum_{j\neq i} \Big(1-\gamma_{i,j}\cdot \gamma_{j,i}\Big)
&(\dme{\text{using }\beta_{i,j}=\gamma_{i,j}\cdot\gamma_{j,i}/n})
%(\text{replacing $\beta_{i,j}$ with its expression})
\\
&\leq 1+\frac{1}{n}\sum_{i=1}^n\sum_{j\neq i} \Big((1-\gamma_{i,j}) + (1-\gamma_{j,i})\Big)
&(\text{because $\gamma_{i,j}+\gamma_{j,i} \leq 1+\gamma_{i,j}\cdot\gamma_{j,i}$})\\
&\leq 1+\frac{2}{n}\sum_{i=1}^n\sum_{j\neq i} (1-\low{i}{j}/\val_i(\sigs)).
&(\text{because $\low{i}{j}/\val_i(\sigs) \leq \gamma_{i,j}$})
\end{align*}
\begin{align*}
B &= \sum_{i=1}^n \sum_{j\neq i} \frac{\gamma_{i,j}\cdot\gamma_{j,i}}{n} \cdot (\ln(\val_i(\sigs))-\ln(\low{i}{j}))&
\dme{(\text{using }\beta_{i,j}=\gamma_{i,j}\cdot\gamma_{j,i}/n)}
%\text{(replacing $\beta_{i,j}$ with its expression)}
\\
&\leq \frac{1}{n}\sum_{i=1}^n \sum_{j\neq i} \gamma_{i,j} \cdot (\ln(\val_i(\sigs))-\ln(\low{i}{j}))
&\text{(using that $\gamma_{j,i}\leq 1$)}\\
&\leq \frac{1}{n}\sum_{i=1}^n \sum_{j\neq i}(1-\low{i}{j}/\val_i(\sigs)).
&\text{(replacing $\gamma_{i,j}$ with its expression)}
\end{align*}
Overall, \dme{by adding together $A$ and $B$} we obtained that
$$
\sum_{i=1}^n y_i \leq 1 + \frac{3}{n}\sum_{i=1}^n\sum_{j\neq i} (1-\low{i}{j}/\val_i(\sigs)).
$$
Using \Cref{lem:self-bounding}, this last expression is upper-bounded by $4$. 
\mfe{Recalling that, by Lemma~\ref{lem:eating-lp}, $y_i$ is the pre-normalized share, and that $x_i$ is set to $y_i/4$, this implies that $\sum_{i=1}^n x_i \le 1$, proving feasibility.}
\end{proof}

%%%%%%%%%%%%%%%%%%%%%%%%%%%%%%%%%%%%%%%%
\section{Optimal Mechanism for Matroid Feasibility Constraints and $d$-Critical Valuations} \label{sec:matroids}

In this section, we devise a mechanism for $d$-critical valuations within settings where the feasible sets of bidders to be served form a matroid. \dme{We show that our mechanism obtains a $(d+1)$-approximation to the optimal welfare truthfully. This bound is optimal even for single-item settings. The main result of this section follows.
}

\begin{theorem}\label{thm:d-critical-matroid}
    For every setting with bidders with $d$-critical valuations, where the feasible sets of bidders form a matroid, 
    Mechanism \ref{alg:RCP} is an EPIC-IR mechanism that gives a $(d+1)$-approximation to the social welfare.
\end{theorem}

Our mechanism first finds a set of potential \textit{candidates} $\mathcal{C}$, comprising the only bidders eligible for allocation. 
To determine whether bidder $i$ is a candidate, the following process takes place (for each bidder $i$ separately): the mechanism considers bidder $i$'s \textit{real} value alongside the low estimates of all other bidders $j \neq i$.
Bidders are sorted in non-increasing order of these values (breaking ties in favor of low index), and are added greedily to a set $I_i$ \dme {provided} an independent set is maintained. Bidder $i$ belongs to the set of candidates $\cands$ if $i$ belongs to $I_i$.

We first show that every bidder that belongs to the optimal solution (i.e., the welfare-maximizing solution based on the real values) also belongs to $\cands$.
Therefore, to establish a $(d+1)$-approximation, it suffices to show that every bidder in $\cands$ is served with probability  $1/(d+1)$. 
The set $\cands$, however, need not be an independent set in the matroid.
We then show that the condition of \Cref{thm:partition} holds when restricting attention to the bidders in $\cands$, for $t=d+1$, and thus $\mathcal{C}$ can be partitioned into at most $d+1$ independent sets. It follows that, by choosing each one of these $\le d+1$ independent sets with probability $1/(d+1)$, every bidder in $\cands$, and thus also every bidder in the optimal solution, is served with probability $1/(1+d)$, yielding the desired approximation guarantee (which is tight according to \Cref{prop:d_critical_lb}).

\begin{algorithm}[h]
\textbf{Function} $\texttt{greedy}(\text{weight function }w:[n]\rightarrow \mathbb R_+):$
\begin{enumerate}
\item Initialize $I=\emptyset$.
\item Sort bidders $i$ by decreasing weights $w(i)$, and let $\sigma$ be the resulting ordering, breaking ties in favor of lower-index bidders.
\item Go over bidders according to ordering $\sigma$, for a current bidder $j$, add $j$ to $I$ if $I\cup \{j\}\in \indep$.
\item Return the set $I$.
\end{enumerate}
\textbf{Mechanism} \texttt{CP}:
\begin{enumerate}
\item Elicit signals $\sigs[\hat]$ and valuation functions $\val[\hat]$. Initialize the set of candidates $\cands=\emptyset$.
\item For each bidder $i$:\\
\quad Define the function $w_i$ where $w_i(i) = \val[\hat]_i(\sigs[\hat])$ and $w_i(j)=\low[\hat]{j}{i}$ for $j \neq i$.\\
\quad Define $I_i \gets \texttt{greedy}(w_i)$, then add $i$ to $\cands$ if $i \in I_i$.
\item Find a partitioning of $\cands$ using at most $d+1$ independent sets $\indep_\cands$ (which exists by \Cref{lem:candidate_partitioning}).   
\item Serve each set $I\in \indep_\cands$ with probability $1/(d+1)$.
\item Charge payments using Eq~\eqref{eq:pricing}.
\end{enumerate}
\renewcommand{\thealgocf}{CP}
\renewcommand{\algorithmcfname}{Mechanism}
\caption{\textbf{Candidate Partitioning (CP) Mechanism}.}
\label{alg:RCP}
\end{algorithm}

\dme{Towards proving \Cref{thm:d-critical-matroid}}, 
we first show that the mechanism is truthful.

\begin{lemma} \label{lem:rpc_truthful}
    Mechanism \ref{alg:RCP} is EPIC-IR.
\end{lemma}
\begin{proof}
    By \Cref{lem:pricing}, it suffices to show that the allocation depends only on $\sigs[\hat]_{-i}$, $\vals[\hat]_{-i}$, and $\val[\hat]_i(\sigs[\hat])$, and is monotone non-decreasing in $\val[\hat]_i(\sigs[\hat])$. We first observe that if $i$ is a candidate, then $i$ is allocated with probability $1/(d+1)$, regardless of the composition of $\cands$. To determine whether $i$ is a candidate, we consider $w_i(i)$'s relative position compared to every other $w_i(j)$, where these values depend only on $\sigs[\hat]_{-i}$, $\vals[\hat]_{-i}$, and $\val[\hat]_i(\sigs[\hat])$. An increase in $\val[\hat]_i(\sigs[\hat])$ boosts $w_i(i)$ without affecting all other $w_i(j)$ values. Thus, by \Cref{lemma:before_still_selected}, if $i$ is added to $I_i$ for some $\val[\hat]_i(\sigs[\hat])$, it is also added when $\val[\hat]_i(\sigs[\hat])$ is increased, establishing monotonicity. This implies the EPIC-IR property of the mechanism. 
\end{proof}

Since the mechanism is EPIC-IR, hereafter, we use $\vals$ and $\sigs$ instead of $\vals[\hat]$ and $\sigs[\hat]$.

Let $\sigma$ be the ordering of the bidders according to their \textit{actual values} $v_i(\sigs)$, with ties broken in favor of lower-index bidders. Let $I^\star$ be the independent set returned by the greedy algorithm described above, iterating over bidders in the order of $\sigma$ (and adding bidder $i$ into $I^\star$ if $I^\star\cup {i} \in \indep$). By Theorem~\ref{thm:greedy}, $I^\star$ is a welfare-maximizing feasible set of bidders, i.e., $\sum_{i\in I^\star} v_i(\sigs) = \max_{I\in \indep} \sum_{i\in I} v_i(\sigs)$.

The following lemma shows that all bidders in $I^\star$ will be candidates. 

\begin{lemma} \label{lem:opt_candidates}
    Let $I^\star$ be the welfare-maximizing set computed by the greedy algorithm (iterating over bidders according to the order $\sigma$). It holds that $I^\star\subseteq \cands$.
\end{lemma}

\begin{proof}
    Consider a bidder $i\in I^\star$. Consider the weight functions $w(j)=v_j(\sigs)$ for all $j$, and $\hat{w} = w_i$ as defined in Mechanism~\ref{alg:RCP}. Notice that $w(i)=v_i(\sigs)=\hat{w}(i)$ and that $w(j)=v_j(\sigs)\ge \low{j}{i}=\hat{w}(j)$ for all $j\neq i$. Notice also that $I^\star$ and $I_i$ are the independent sets resulting from the greedy algorithm using weight functions $w$ and $\hat{w}$, respectively, breaking ties in favor of lower-index bidders. Therefore, by $\Cref{lemma:before_still_selected}$, $i\in I_i$, which implies that $i\in \cands$.
\end{proof}

We next show that there exists a partitioning of $\mathcal{C}$ into at most $d+1$ independent sets.

\begin{lemma} \label{lem:candidate_partitioning}
    There exists a set of $k\le d+1$ disjoint independent sets $\indep_{\cands}=\{I_1,\ldots, I_{k}\}$ such that $I_1 \uplus \dots \uplus I_{k}=\cands$.
\end{lemma}

\begin{proof}
    We show that Eq.~\eqref{eq:partition_condition} holds for $\cands$ with $t=d+1$, namely, that 
    $|S| \leq  (d+1)\cdot rank(S)$ for every $S \subseteq \cands$.
    We first show that Eq.~\eqref{eq:partition_condition} holds for $S=\cands$. Since $I^\star\subseteq \cands$ (by \Cref{lem:opt_candidates}), $rank(\cands)=rank(M) = |I^\star|$. Thus, to show  that Eq.~\eqref{eq:partition_condition} holds, it suffices to show that $|\cands|\le (d+1)\cdot|I^\star|.$ For $i\notin I^\star$ to be a candidate, there must be at least one $j\in I^\star$ such that $\low{j}{j} < v_j(\sigs)$. Otherwise, by \Cref{lemma:before_still_selected}, all bidders in $I^\star$ will be added to $I_i$ when iterating over the bidders in order $\sigma_i$, and since $rank(I^\star)=rank(M)$, bidder $i$ would not be selected. 

    For $i\in I^\star$, let $D_i$ be the set of bidders for which $\low{j}{j} < v_j(\sigs)$. 
    Since $v_i$ is $d$-critical, $|D_i|\le d$. By the above argument, $\cands\setminus I^\star \subseteq \cup_{i\in I^\star} D_i$.
    Therefore, 
    $$|\cands| = |\cands\setminus I^\star|+|I^\star|\le |\cup_{i\in I^\star} D_i| + |I^\star| \le d\cdot |I^\star|+|I^\star| = (d+1)\cdot|I^\star|,$$
    where the second inequality follows from the union bound. Therefore, Eq.~\eqref{eq:partition_condition} holds for the set $\cands$. 

    Now consider a set of candidates $S\subseteq \cands$. Since the rank function is submodular with marginal contribution at most 1, $$rank(\cands)-rank(S)\le |\cands| - |S|.$$
    Therefore, 
    $$\frac{|S|}{rank(S)}\le \frac{|S| + (|\cands| - |S|)}{rank(S) + (rank(\cands)-rank(S))} = \frac{|\cands|}{rank(\cands)} \le d+1,$$
    satisfying Eq.~\eqref{eq:partition_condition}.
    %as desired by \Cref{thm:partition}. 
    It follows by \Cref{thm:partition} that one can find a partition of $\mathcal{C}$ into at most $d+1$ independent sets, as desired.
\end{proof}

We are now ready to prove the main theorem of this section. 

\begin{proof}[Proof of \Cref{thm:d-critical-matroid}]
    By \Cref{lem:rpc_truthful}, the mechanism is EPIC-IR. By \Cref{lem:candidate_partitioning}, the mechanism is well-defined, and always outputs a feasible set of bidders. It remains to show that the mechanism gives a $(d+1)$-approximation. Consider the welfare-maximizing set of bidders, $I^\star$, which is the output of the greedy algorithm when iterating over the bidders according to their real values. By \Cref{lem:opt_candidates}, $I^\star\subseteq\cands$. 
    Let $\indep_{\cands}$ be a set of $k\le d+1$ independent sets forming a partition of $\cands$ (which exists by \Cref{lem:candidate_partitioning}). 
    For every $i\in I^\star$, let $\indep_\cands(i)$ denote the set in $\indep_\cands$ containing $i$. We have that the expected welfare is
    \begin{eqnarray*}
        \sum_{i=1}^n v_i(\sigs)\cdot\Pr[i\mbox{ is allocated}] & \ge & \sum_{i\in I^\star} v_i(\sigs)\cdot\Pr[i\mbox{ is allocated}] \\ 
        & = & \sum_{i\in I^\star} v_i(\sigs)\cdot\Pr[\indep_\cands(i)\mbox{ is served}] \\
        & = & \frac{1}{d+1} \sum_{i\in I^\star} v_i(\sigs) \\
        & = & \frac{OPT}{d+1},
    \end{eqnarray*}
    as desired. This concludes the proof of the theorem.
\end{proof}

\textbf{Remark.} We note that the \ref{alg:RCP} Mechanism can be applied to non-monotone critical valuations, by changing the definition of $i$'s low estimate for $j$'s value to be $\inf_{o_i\in S_i}v_j(o_i, \sigs_{-i})$.

\subsection{Bidders with Heterogeneous and Unknown $d$}
Similarly to~\citet{EdenGZ22,EdenFGMM23}, we can extend the results to the case where each bidder $i$'s valuation is $d_i$-critical with a different $d_i$ and $d_i$ is private information. We adjust the mechanism as follows.

\newcommand{\imax}[0]{\tilde{i}}

In addition to answering value queries, each bidder $i$ also reports $\hat{d}_i$. For every $i$, we compute $\bar{d}^i=\max_{j\neq i} \hat{d}_j$. Note that $(i)$ $\bar{d}^i$ only depends on the reports of bidders $j\neq i$; and $(ii)$ assuming bidders bid truthfully, all bidders face the same $\bar{d}$, except (potentially) one bidder --- the bidder $i$ with the maximum overall $d_{i}$ value, call this bidder $\imax$. We then find an allocation rule such that if a bidder $i\neq \imax$ becomes a candidate, they are served with probability $1/2(\bar{d}+1)$, and if $\imax$ is a candidate, $\imax$ is served with probability $1/2(\bar{d}^{\imax}+1)$. This ensures truthfulness as the probability a bidder is served depends on the reports of other bidders, and is monotone in $\hat{v}(\hat{\sigs})$ (since by \Cref{lemma:before_still_selected}, being a candidate is monotone in $\hat{v}(\hat{\sigs})$). 

It remains to present an allocation rule that satisfies the above conditions. This is done as follows.
We flip a fair coin. If it turns head, then if $\imax$ is a candidate, we  serve $\imax$ with probability $\frac{1}{\bar{d}^{\imax}+1}$. 
If it turns tails, then if $\imax \in \cands$, we remove $\imax$ from $\cands$. 
Since all bidders are $\bar{d}$-critical, we can apply the same argument as in~\Cref{lem:candidate_partitioning}, and partition $\cands$ into at most $\bar{d}+1$ independent sets, choosing each one with probability $\frac{1}{\bar{d}+1}$.
To show that this mechanism gives $2(\bar{d}+1)$-approximation, note that by \Cref{lem:opt_candidates}, every bidder in $I^\star$ is a candidate, and is served with probability at least $1/2(\bar{d}+1)$.

% Bibliography
\bibliographystyle{plainnat}
\bibliography{interdependent.bib}

% Appendix
\appendix

\end{document}